\documentclass[11pt]{article}
\usepackage[a4paper,margin=29mm]{geometry}
\usepackage[T1]{fontenc}
\usepackage{lmodern,microtype}
\usepackage{amsmath,amssymb,amsthm,mathtools}
\usepackage{aliascnt}
\usepackage[numbers,sort&compress]{natbib}
\usepackage[hidelinks]{hyperref}
\usepackage[nameinlink,capitalise,noabbrev]{cleveref}
\hypersetup{
  pdftitle={Variable-Cliff Nielsen Geometry and an Exponent-4/3 Lower Bound for the Infinite-Cliff Diameter},
  pdfauthor={Honghuai Fang},
  pdfsubject={Variable-cliff Nielsen geometry, infinite-cliff diameter bounds, and projective circuit complexity},
  pdfkeywords={quantum complexity geometry, infinite cliff, Carnot--Caratheodory distance, Jacobi fields, Weyl integration, Haar measure, approximate circuit complexity}
}
\allowdisplaybreaks

\theoremstyle{plain}
\newtheorem{theorem}{Theorem}[section]
\newaliascnt{proposition}{theorem}
\newtheorem{proposition}[proposition]{Proposition}
\aliascntresetthe{proposition}
\newaliascnt{lemma}{theorem}
\newtheorem{lemma}[lemma]{Lemma}
\aliascntresetthe{lemma}
\newaliascnt{corollary}{theorem}
\newtheorem{corollary}[corollary]{Corollary}
\aliascntresetthe{corollary}
\crefname{theorem}{theorem}{theorems}
\Crefname{theorem}{Theorem}{Theorems}
\crefname{proposition}{proposition}{propositions}
\Crefname{proposition}{Proposition}{Propositions}
\crefname{lemma}{lemma}{lemmas}
\Crefname{lemma}{Lemma}{Lemmas}
\crefname{corollary}{corollary}{corollaries}
\Crefname{corollary}{Corollary}{Corollaries}

\newcommand{\PU}{\operatorname{PU}}
\newcommand{\Vol}{\operatorname{Vol}}
\newcommand{\Tr}{\operatorname{Tr}}
\newcommand{\diam}{\operatorname{diam}}
\newcommand{\spec}{\operatorname{spec}}
\newcommand{\ad}{\operatorname{ad}}
\newcommand{\Ad}{\operatorname{Ad}}
\newcommand{\op}{\mathrm{op}}
\newcommand{\HS}{\mathrm{HS}}
\newcommand{\Sone}{S_1}
\newcommand{\sinc}{\operatorname{sinc}}
\newcommand{\E}{\mathbb E}
\newcommand{\Prob}{\mathbb P}
\newcommand{\cL}{\mathcal L}
\newcommand{\cP}{\mathcal P}
\newcommand{\C}{\mathbb C}
\newcommand{\rank}{\operatorname{rank}}
\newcommand{\sym}{\operatorname{sym}}
\newcommand{\Arg}{\operatorname{Arg}}
\newcommand{\diag}{\operatorname{diag}}
\newcommand{\vol}{\operatorname{vol}}
\newcommand{\Vdm}{\operatorname{Vdm}}

\title{Variable-Cliff Nielsen Geometry and an Exponent-$4/3$ Lower Bound
for the Infinite-Cliff Diameter}
\author{Honghuai Fang\thanks{Institute for Theoretical Sciences, Westlake University,
600 Dunyu Road, Xihu District, Hangzhou, Zhejiang 310030, China.
Email: \texttt{fanghonghuai@westlake.edu.cn}.}}
\date{}

\begin{document}
\maketitle

\begin{abstract}
Let $D=2^n$ and $M=3n+9\binom n2$.  We study the right-invariant
one-step-cliff metric $d_Q$ on $\PU(D)$, with unit penalty on Pauli weights one
and two and penalty $Q$ on all higher weights.  If
$M/Q_D\to0$ and $MQ_D^{3/4}/D^2\to0$, then, for every fixed
$0<x<\pi/\sqrt3$,
\[
 \mu_D\bigl(B_{Q_D}([I],x\sqrt{Q_D})\bigr)\le e^{-c_xD^2}.
\]
Here $\mu_D$ is normalized Haar measure.  Thus the Haar-typical distance
from the identity and the diameter are both asymptotic to
$(\pi/\sqrt3)\sqrt{Q_D}$ throughout the window
$M\ll Q_D\ll D^{8/3}M^{-4/3}$.  Choosing
$Q_D=\kappa D^{8/3}M^{-4/3}$ with sufficiently small fixed $\kappa>0$
yields a Haar-typical lower bound of order $D^{4/3}M^{-2/3}$ for the
corresponding infinite-cliff Carnot--Carath\'eodory distance, outside an
$e^{-\Omega(D^2)}$ exceptional set.  The infinite-cliff diameter therefore
has exponential lower rate at least $4/3$, disproving Brown's exponent-one
conjecture.  The same estimate gives a fixed-error no-ancilla two-qubit
circuit lower bound of the same order.
\end{abstract}



\section{Introduction}\label{sec:intro}

Nielsen's geometric approach replaces a quantum circuit by a path in a
unitary group and assigns larger metric cost to Hamiltonian directions with
larger many-body support
\cite{NielsenQIC,NielsenScience,NielsenPRA,DowlingNielsen}.  The resulting
right-invariant geometry permits arbitrary time-dependent local Hamiltonians,
so its lower bounds are not reducible to parameter counting for exact gate
synthesis \cite{Shende}.

For $n$ qubits, write $D=2^n$.  Brown obtained exponential lower bounds for
Haar-typical Nielsen complexity by Bishop--Gromov comparison
\cite{BrownBG}.  For the one-step cliff, with unit penalty on Pauli weights
one and two and common penalty $Q$ on all higher weights, those estimates
determine the $\sqrt Q$ scale at the level of exponential rates in $n$
for $Q<D^2$ \cite[Eq.~(2.7) and Fig.~1]{BrownBG}.  Brown further conjectured
that the limiting infinite-cliff metric, in which weights at least three are
forbidden, has diameter $D$ at exponential scale
\cite[Eq.~(3.2)]{BrownBG}.

The companion manuscript \cite{FangFixed} treats the height $Q=D^2$.
It proves the sharp Haar-typical distance, two-sided small-ball bounds
below the threshold at speed $D^2$, and the near-threshold coefficient
$1/(16\zeta(3))$.  We use its determinant comparison, comparison of
spherical averages, Weyl integration formula, and Abel estimate for the
logarithmic kernel.

The main task is to control the Jacobi comparison uniformly as the cliff
height varies.  The two leading losses are
\[
 MQ_D^{3/4}
 \qquad\text{and}\qquad
 D^2\sqrt{\frac{M}{Q_D}},
\]
where $M=3n+9\binom n2$ is the dimension of the weight-one-and-two Pauli
subspace.  The first term comes from rescaling the Jacobi equations by
quarter powers of the inverse inertia on time intervals of length
$O(\sqrt Q)$; the second comes from removing the low-weight component of
the normalized momentum.  Requiring both losses to be
$o(D^2)$ gives
\[
 M\ll Q_D\ll D^{8/3}M^{-4/3}.
\]
Throughout this window, the Haar-typical distance and the diameter are
asymptotic to $(\pi/\sqrt3)\sqrt{Q_D}$.

At the upper edge of the window, take
$Q_D=\kappa D^{8/3}M^{-4/3}$ with $\kappa>0$ sufficiently small.  Penalty
monotonicity then transfers the finite-cliff small-ball estimate to the
Carnot--Carath\'eodory limit and gives a Haar-typical lower bound of order
$D^{4/3}M^{-2/3}$.  Consequently the infinite-cliff diameter has exponential
lower rate at least $4/3$, contradicting the exponent-one conjecture.  The
same critical-height argument gives a fixed-error no-ancilla two-qubit
circuit lower bound of the same order outside an $e^{-\Omega(D^2)}$
exceptional set.

The proof has three $Q$-dependent stages.  First, in spatial coordinates the
inverse inertia is
\[
 G(t)=P(t)+Q^{-1}P(t)^\perp,
\]
where $P(t)$ is a moving rank-$M$ projection.  Opposite quarter powers of
$G(t)$ balance the motion of the low-weight subspace against the
commutator terms coupling it to its orthogonal complement.  This gives
the integrated error $O(MQ^{3/4})$.  Second, the metric polar Jacobian is
normalized so that all global powers of $Q$ cancel
after the radial scaling $t=s\sqrt Q$.  Third, the normalized momentum is
reduced to its high-weight component at cost
$O(D^2\sqrt{M/Q})$.

A projective phase average also gives the deterministic upper bound
\[
 \diam(\PU(D),d_Q)\le\frac{\pi}{\sqrt3}\sqrt Q
\]
for every finite $Q$.  Thus the lower and upper finite-cliff scales agree
throughout the variable-height window.

Brown later established polynomial equivalence for broad classes of
complexity geometries \cite{BrownPoly}, and long-distance universality for
several penalty schedules was studied in \cite{BrownUniversality}.  The
Euler--Arnold and Jacobi equations for several cost factors were developed in
\cite{RibeiroTrancanelli}.  The spectral bound along minimizing geodesics
follows from the conjugate-time estimate of Le Brigant, Lichtenfelz, and
Preston \cite[Corollary~3.5]{LBLP}.

\section{Geometric framework and main results}\label{sec:setup}

\subsection{Cliff metrics and projective complexity}

Let $D=2^n$, and set
\[
 \mathfrak h=\{H\in M_D(\C):H^\dagger=H,\ \Tr H=0\},
 \qquad
 \mathfrak g=i\mathfrak h=\mathfrak{su}(D)\cong\mathfrak{pu}(D).
\]
We identify these real Hilbert spaces through $X\mapsto-iX$ and use the
normalized Hilbert--Schmidt inner product
\begin{equation*}
 \langle X,Y\rangle_0=D^{-1}\Tr(X^\dagger Y),
 \qquad X,Y\in\mathfrak g.
\end{equation*}
Let $\{\sigma_I\}$ denote the Hermitian $n$-qubit Pauli strings, normalized by
$D^{-1}\Tr(\sigma_I\sigma_J)=\delta_{IJ}$.  Then
$e_I=i\sigma_I$, for $I\ne0$, is an orthonormal skew-Hermitian basis of
$\mathfrak g$.  The weight of $e_I$ is the number of nonidentity tensor
factors in $\sigma_I$.  Let $P$ be the orthogonal projection on
$\mathfrak g$ onto weights one and two; through $X\mapsto-iX$ we use the same
letter for the corresponding projection on $\mathfrak h$.  Put
\begin{equation*}
 M=\rank P=3n+9\binom n2=\frac92n^2-\frac32n,
 \qquad N=D^2-1.
\end{equation*}
Write $\cL=\operatorname{ran}P\subset\mathfrak g$ and
$\cP=\operatorname{ran}P^\perp\subset\mathfrak g$, and write
$\mathfrak h_{\cL}=-i\cL$ and $\mathfrak h_{\cP}=-i\cP$ on the Hermitian
side.

For $Q\ge1$ define the inertia and inverse inertia
\begin{equation*}
 A_Q=P+QP^\perp,
 \qquad G_Q=A_Q^{-1}=P+qP^\perp,
 \qquad q=Q^{-1}.
\end{equation*}
The inner product
$\langle X,Y\rangle_Q=\langle X,A_QY\rangle_0$ is extended right
invariantly to $\PU(D)$, and we write
$\|X\|_Q=\langle X,X\rangle_Q^{1/2}$.  Denote the corresponding distance by
$d_Q$ and the Riemannian volume by $\Vol_Q$.  For $p\in\PU(D)$ and $r>0$, set
\begin{equation*}
 B_Q(p,r)=\{y\in\PU(D):d_Q(p,y)<r\}.
\end{equation*}
For an absolutely continuous curve $\gamma:[a,b]\to\PU(D)$, set
\begin{equation*}
 \ell_Q(\gamma)=\int_a^b\|\dot\gamma(t)\|_Q\,dt.
\end{equation*}
Since every right-invariant volume on a compact group is a scalar multiple
of Haar measure, the normalized measure
\begin{equation*}
 \mu_D(E)=\frac{\Vol_Q(E)}{\Vol_Q\PU(D)}
\end{equation*}
is independent of $Q$.

The infinite-cliff distance $d_\infty$ is the Carnot--Carath\'eodory distance
of the right-invariant horizontal distribution generated by $\cL$, with
horizontal norm $\|X\|_0$.  Equivalently, it is the infimum of
$\int_0^1\|u(t)\|_0\,dt$ over absolutely continuous curves whose
right-trivialized velocity belongs to $\cL$ almost everywhere; see
\cite{Montgomery,ABB} for the standard sub-Riemannian framework.

For projective classes define
\begin{equation*}
 \delta_{\op}([U],[V])=
 \inf_{\phi\in\mathbb R}\|U-e^{i\phi}V\|_{\op}.
\end{equation*}
For $0\le\epsilon<2$, let
$\mathcal C^{\rm proj}_\epsilon([U])$ be the least number of arbitrary
two-qubit gates in a no-ancilla circuit, with arbitrary one-qubit gates free,
whose endpoint $V$ satisfies
$\delta_{\op}([U],[V])\le\epsilon$.

\subsection{Jacobian conventions and basic comparisons}

Determinants and Jacobians in the geometric argument are taken over real
tangent spaces.  After complexifying the commutator operators, each
unordered root pair $\{i,j\}$ represents a real two-plane and therefore
contributes twice.  In the abstract linear-algebra statements,
$\mathbb F$ denotes either $\mathbb R$ or $\mathbb C$; dimensions and
singular-value multiplicities are taken over $\mathbb F$.

If $L:E\to\mathcal H$ is linear with $k$-dimensional domain, write
\begin{equation*}
 \vol_k(L)=\det(L^*L)^{1/2}=\prod_{j=1}^k s_j(L),
\end{equation*}
and, for $\delta>0$, define
\begin{equation*}
 \vol_{k,\delta}(L)=\prod_{j=1}^k\max\{s_j(L),\delta\}.
\end{equation*}
For a square map at radial parameter $s$ and relative cutoff $0<\rho\le1$,
write
\begin{equation*}
 \det_{\rho,s}L=\prod_j\max\{s_j(L),\rho s\}.
\end{equation*}

For a linear map $T$ between finite-dimensional real or complex Hilbert
spaces, we denote its operator, Hilbert--Schmidt, and trace norms by
\[
 \|T\|_{\op},\qquad \|T\|_{\HS},\qquad \|T\|_{\Sone},
\]
respectively, and set
\[
 \sym T:=\frac12(T+T^*).
\]

Throughout,
\[
 \sinc t:=
 \begin{cases}
  \sin t/t,&t\ne0,\\
  1,&t=0,
 \end{cases}
 \qquad
 \log_+x:=\max\{\log x,0\}\quad(x>0).
\]
Finally,
\[
 \omega_k:=\frac{\pi^{k/2}}{\Gamma(k/2+1)}
\]
denotes the Euclidean volume of the unit ball in $\mathbb R^k$; hence the
area of $S^{k-1}$ is $k\omega_k$.

\begin{lemma}\label{lem:bracket-generation}
The Lie algebra generated by the Pauli strings of weights one and two is
$\mathfrak{su}(D)$.  Consequently the right-invariant distribution $\cL$ on
$\PU(D)$ is bracket generating and $d_\infty$ is finite.
\end{lemma}

\begin{proof}
The skew-Hermitian matrices $e_I=i\sigma_I$, $I\ne0$, span
$\mathfrak g$.  We prove by induction on the weight that every $e_I$ belongs
to the real Lie algebra generated by weights at most two.  The assertion is
immediate for weights one and two.  Let the underlying Hermitian Pauli string
$S$ have weight $k\ge3$ and choose distinct sites $a,b$ in its support.  At
site $b$, choose distinct Pauli matrices $R_b,T_b$ whose product is a nonzero
imaginary scalar multiple of the factor of $S$ at $b$.  Let $R$ agree with
$S$ away from $a,b$, have factor $R_b$ at $b$, and be the identity at $a$.
Let $T$ be supported on $\{a,b\}$, have at $a$ the factor of $S$, and have
factor $T_b$ at $b$.  Then $R$ has weight $k-1$, $T$ has weight two, and they
anticommute at exactly one site.  Hence
\[
 [iR,iT]=-[R,T]
\]
is a nonzero real multiple of $iS$.  The induction closes.  The
Chow--Rashevskii theorem now gives finiteness of the associated
Carnot--Carath\'eodory distance \cite{Montgomery,ABB}.
\end{proof}

\begin{lemma}\label{lem:monotone}
If $1\le Q_1\le Q_2$, then
\begin{equation}\label{eq:monotone}
 d_{Q_1}([U],[V])\le d_{Q_2}([U],[V])\le d_\infty([U],[V]).
\end{equation}
\end{lemma}

\begin{proof}
The pointwise norm for $Q_2$ dominates that for $Q_1$.  Every horizontal path
has zero $P^\perp$ component and therefore has the same length for every
finite $Q$.  Taking infima proves both inequalities.
\end{proof}

\subsection{Main results}

Set
\begin{equation*}
 x_\star=\frac{\pi}{\sqrt3}
\end{equation*}
and define
\begin{equation}\label{eq:Xi}
 \Xi_{D,Q}:=
 \frac{MQ^{3/4}}{D^2}
 +\sqrt{\frac{M}{Q}}
 +\frac{M}{D^{3/2}}
 +\frac{\log D}{D}.
\end{equation}

\begin{theorem}[Variable-height small balls]
\label{thm:variable-fixed}
Let $Q_D\ge1$ satisfy
\begin{equation}\label{eq:variable-window-assumptions}
 \frac{MQ_D^{3/4}}{D^2}\longrightarrow0,
 \qquad
 \frac{M}{Q_D}\longrightarrow0.
\end{equation}
For every fixed $0<x<x_\star$, there are constants $c_x>0$ and $D_x$ such
that
\begin{equation}\label{eq:variable-fixed-ball}
 \mu_D\bigl(B_{Q_D}([I],x\sqrt{Q_D})\bigr)
 \le e^{-c_xD^2}
\end{equation}
for every power of two $D\ge D_x$.
\end{theorem}

\begin{theorem}[Sharp leading-order variable-height asymptotics]\label{thm:variable-sharp}
If $\Xi_{D,Q_D}\to0$, there are absolute constants $c,C>0$ such that,
for all sufficiently large powers of two $D$,
\begin{equation}\label{eq:variable-concentration}
 \mu_D\left\{[U]:
 \frac{d_{Q_D}([I],[U])}{\sqrt{Q_D}}
 <x_\star-C\Xi_{D,Q_D}^{1/4}\right\}
 \le\exp\{-c\Xi_{D,Q_D}^{1/2}D^2\}.
\end{equation}
Moreover every $[U]\in\PU(D)$ satisfies
\begin{equation}\label{eq:det-upper-main}
 d_{Q_D}([I],[U])\le x_\star\sqrt{Q_D}.
\end{equation}
Consequently, if $U_D$ is Haar distributed on $\PU(D)$, then
\begin{equation}\label{eq:variable-prob-limit}
 \frac{d_{Q_D}([I],[U_D])}{\sqrt{Q_D}}
 \longrightarrow x_\star
\end{equation}
in probability and in every fixed finite $L^p$, and
\begin{equation}\label{eq:variable-diameter}
 \frac{\diam(\PU(D),d_{Q_D})}{\sqrt{Q_D}}
 \longrightarrow x_\star.
\end{equation}
\end{theorem}

At the distinguished height $Q_D=D^2$, \eqref{eq:Xi} becomes
\[
 \Xi_{D,D^2}
 =\frac{M}{D^{1/2}}+\frac{\sqrt M}{D}
  +\frac{M}{D^{3/2}}+\frac{\log D}{D}
 =(1+o(1))\frac{M}{D^{1/2}}.
\]
Consequently,
\[
 \Xi_{D,D^2}^{1/4}
 =(1+o(1))\left(\frac{M}{D^{1/2}}\right)^{1/4},
 \qquad
 \Xi_{D,D^2}^{1/2}D^2
 =(1+o(1))M^{1/2}D^{7/4}.
\]
Thus \cref{thm:variable-sharp} recovers the concentration window and
lower-tail exponent of \cite[Theorem~2.1]{FangFixed}.  The two-sided
small-ball bounds and the near-threshold coefficient $1/(16\zeta(3))$ at
this height are proved in \cite[Theorem~2.2 and Corollary~2.3]{FangFixed};
fluctuations of the centered principal-logarithm path are described in
\cite[Proposition~A.3]{FangFixed}.

\begin{theorem}[Infinite-cliff lower bound]
\label{thm:infinite-main}
There are absolute constants $c_0,c_1>0$ such that, for every sufficiently
large power of two $D=2^n$,
\begin{equation}\label{eq:robust-infinite}
 \mu_D\left\{[U]:d_\infty([I],[U])\ge
 c_0\frac{D^{4/3}}{M^{2/3}}\right\}
 \ge1-e^{-c_1D^2}.
\end{equation}
In particular,
\begin{equation}\label{eq:robust-diameter}
 \diam(\PU(D),d_\infty)\ge
 c_0\frac{D^{4/3}}{M^{2/3}},
\end{equation}
and hence
\begin{equation}\label{eq:exp-liminf}
 \liminf_{n\to\infty}\frac1n\log_2
 \diam\bigl(\PU(2^n),d_\infty\bigr)\ge\frac43.
\end{equation}
\end{theorem}

By \cref{lem:brown-normalization}, the same exponential lower rate holds in
Brown's normalized-trace convention.  Since \cite[Eq.~(3.2)]{BrownBG}
conjectures exponent one, \cref{thm:infinite-main} disproves that conjecture.

\begin{corollary}[Fixed-error projective circuit lower bound]
\label{cor:variable-circuit}
Set
\begin{equation}\label{eq:epsilon-star-variable}
 \epsilon_*=2\sin\frac{\pi}{2\sqrt3}.
\end{equation}
For every fixed $0\le\epsilon<\epsilon_*$, there are constants
$a_\epsilon,b_\epsilon>0$ such that, for every sufficiently large power of
two $D=2^n$,
\begin{equation}\label{eq:variable-circuit-lower}
 \mu_D\left\{[U]:
 \mathcal C^{\rm proj}_\epsilon([U])<
 a_\epsilon\frac{D^{4/3}}{M^{2/3}}
 \right\}
 \le e^{-b_\epsilon D^2}.
\end{equation}
\end{corollary}

\begin{theorem}[Refined subcritical transfer to the infinite cliff]
\label{thm:infinite-refined}
There exist absolute constants $c,C>0$ with the following property.  Let
$\Lambda_D\to\infty$ satisfy $\log \Lambda_D=o(\log D)$, set
\begin{equation}\label{eq:Qchoice-refined}
 Q_D=\frac{D^{8/3}}{M^{4/3}\Lambda_D},
\end{equation}
and let $U_D$ be Haar distributed on $\PU(D)$.  Then, for all sufficiently
large powers of two $D$,
\begin{equation}
\label{eq:infinite-refined}
 \Prob\left\{
 d_\infty([I],[U_D])\ge
 \left(x_\star-C\Lambda_D^{-3/16}\right)
 \frac{D^{4/3}}{M^{2/3}\sqrt{\Lambda_D}}
 \right\}
 \ge 1-\exp\{-cD^2\Lambda_D^{-3/8}\}.
\end{equation}
\end{theorem}

\section{Jacobi comparison and polar reduction}\label{sec:jacobi-polar}

We use the linearized geodesic equations and Kato transport from
\cite[Sec.~3]{FangFixed} to estimate the Jacobi determinants for variable $Q$.

\subsection{Conjugate times and the spatial Jacobi equations}\label{sec:phase-spatial}

\begin{proposition}\label{prop:phase-range}
Let $u_0$ be metric unit, let $c(u_0)$ be its cut time, put
$m_0=A_Qu_0$, and define
\begin{equation}\label{eq:H}
 H=-\frac{i}{\sqrt Q}m_0.
\end{equation}
If $s\sqrt Q<c(u_0)$, then
\begin{equation}\label{eq:phase-range}
 s\,\diam\spec(H)<2\pi.
\end{equation}
\end{proposition}

\begin{proof}
The index-form argument in the proof of
\cite[Proposition~3.1]{FangFixed}, with the bound $A_Q\le QI$, gives
\[
 \tau_1\le\frac{2\pi Q}{\|\ad_{m_0}\|_{\op,0}},
\]
where $\tau_1$ is the first conjugate time.  The denominator is positive
because $m_0\ne0$ and $\mathfrak{su}(D)$ is centerless.
If $h_1,\ldots,h_D$ are the eigenvalues of
$H$, then on a root vector $E_{ij}$,
\[
 \ad_{m_0}E_{ij}=i\sqrt Q\,(h_i-h_j)E_{ij},
\]
so $\|\ad_{m_0}\|_{\op,0}=\sqrt Q\,\diam\spec(H)$.  Since the cut time is no
larger than the first conjugate time \cite{Klingenberg}, the hypothesis
$s\sqrt Q<c(u_0)$ implies $s\sqrt Q<\tau_1$, which proves the claim.
\end{proof}

Write the right-trivialized velocity and momentum as
\[
 u=\dot g g^{-1},\qquad m=A_Qu.
\]
The spatial linearization is obtained as in
\cite[Sec.~3.2]{FangFixed}.  With $O(t)=\Ad_{g(t)^{-1}}$,
$m_0=O(t)m(t)$ is constant.  For
$z=O(t)(\delta g\,g^{-1})$ and $\eta=O(t)\delta m$, one has
\begin{equation*}
 z'=G(t)\eta,
 \qquad
 \eta'=D_0G(t)\eta,
 \qquad D_0=-\ad_{m_0},
\end{equation*}
where
\begin{equation*}
 G(t)=O(t)A_Q^{-1}O(t)^{-1}=P(t)+qP(t)^\perp.
\end{equation*}

\subsection{Pauli trace bounds and quarter-power rescaling}\label{sec:quarter}

\begin{lemma}\label{lem:P-bounds}
For every metric-unit geodesic,
\begin{equation}\label{eq:Pprime}
 \|P'(t)\|_{\Sone}\le4M,
 \qquad
 \|P(t)^\perp D_0P(t)\|_{\Sone}\le2M\sqrt Q.
\end{equation}
\end{lemma}

\begin{proof}
The first estimate is the $Q$-independent part of
\cite[Lemma~3.3]{FangFixed}: for a transported low-weight Pauli basis,
$P'$ is off diagonal, has rank at most $2M$, and satisfies
$\|P'\|_{\HS}^2\le8M$.

For the second estimate,
\[
 \sum_{a=1}^M\|P^\perp D_0e_a\|_0^2
 \le4M\|m_0\|_0^2.
\]
Metric-unit speed gives
\[
 \|m_0\|_0^2=\|A_Qu_0\|_0^2
 \le Q\langle u_0,A_Qu_0\rangle_0=Q.
\]
Hence $\|P^\perp D_0P\|_{\HS}\le2\sqrt{MQ}$; its rank is at most $M$,
so $\|P^\perp D_0P\|_{\Sone}\le2M\sqrt Q$.
\end{proof}

Let
\begin{equation*}
 T(t)=G(t)^{1/4}=P(t)+aP(t)^\perp,
 \qquad a=q^{1/4}=Q^{-1/4},
\end{equation*}
and define
\begin{equation*}
 \zeta=T^{-1}z,
 \qquad \xi=T\eta.
\end{equation*}
Then
\begin{equation}\label{eq:aug-actual}
 \frac d{dt}\binom\zeta\xi=
 \begin{pmatrix}
 -T^{-1}T'&G^{1/2}\\
 0&T'T^{-1}+TD_0GT^{-1}
 \end{pmatrix}\binom\zeta\xi.
\end{equation}
Define $\widehat J(t)v=\zeta(t)$ for the solution with
$\zeta(0)=0$ and $\xi(0)=v$.

Relative to $P(t)\oplus P(t)^\perp$, set
$C_P(t)=P(t)^\perp P'(t)P(t)$ and write
\[
 D_0=\begin{pmatrix}D_{11}&D_{12}\\-D_{12}^*&D_{22}\end{pmatrix},
 \qquad
 P'=\begin{pmatrix}0&C_P^*\\C_P&0\end{pmatrix}.
\]
Direct multiplication gives
\begin{equation*}
 TD_0GT^{-1}=
 \begin{pmatrix}D_{11}&a^3D_{12}\\-aD_{12}^*&qD_{22}\end{pmatrix},
 \qquad
 T'T^{-1}=
 \begin{pmatrix}0&(a^{-1}-1)C_P^*\\(1-a)C_P&0\end{pmatrix}.
\end{equation*}
Define the skew-adjoint reference generator
\begin{equation*}
 K=PD_0P+qP^\perp D_0P^\perp+[P',P]
\end{equation*}
and the augmented reference system
\begin{equation}\label{eq:aug-ref}
 \mathcal K(t)=\begin{pmatrix}0&q^{1/2}I\\0&K(t)\end{pmatrix}.
\end{equation}

\begin{proposition}\label{prop:quarter-action}
Let $\mathcal E$ be the difference between the generators in
\eqref{eq:aug-actual} and \eqref{eq:aug-ref}.  Then
\begin{equation}\label{eq:Epoint}
 \|\mathcal E(t)\|_{\Sone}
 \le C\left[
 q^{1/4}\|P^\perp D_0P\|_{\Sone}
 +q^{-1/4}\|P'\|_{\Sone}+M\right].
\end{equation}
Consequently, uniformly for $0<s\le x$,
\begin{equation}\label{eq:Eint}
 \int_0^{s\sqrt Q}\|\mathcal E(t)\|_{\Sone}\,dt
 \le C_xMQ^{3/4}s.
\end{equation}
\end{proposition}

\begin{proof}
Since $T=aI+(1-a)P$,
\[
 T^{-1}T'=\begin{pmatrix}0&(1-a)C_P^*\\(a^{-1}-1)C_P&0\end{pmatrix},
\]
so $\|T^{-1}T'\|_{\Sone}\le Ca^{-1}\|P'\|_{\Sone}$.  The upper-right
difference is
\[
 G^{1/2}-q^{1/2}I=(1-q^{1/2})P,
\]
which has trace norm at most $M$.  Since
\[
 [P',P]=\begin{pmatrix}0&-C_P^*\\C_P&0\end{pmatrix},
\]
the exact lower-right error is
\[
 \begin{pmatrix}
 0&a^3D_{12}+a^{-1}C_P^*\\
 -aD_{12}^*-aC_P&0
 \end{pmatrix}.
\]
Its trace norm is at most
$C(a\|P^\perp D_0P\|_{\Sone}+a^{-1}\|P'\|_{\Sone})$.  This proves
\eqref{eq:Epoint}.  By \cref{lem:P-bounds}, the pointwise cost is
$O(MQ^{1/4}+M)$; integration over $[0,s\sqrt Q]$ gives \eqref{eq:Eint}.
\end{proof}

The reference propagator is
\begin{equation}\label{eq:Mref}
 \mathcal M_K(t)=
 \begin{pmatrix}
 I&q^{1/2}\int_0^tW(r)\,dr\\0&W(t)
 \end{pmatrix},
\end{equation}
where $W'=KW$ is unitary.  For $t\le x\sqrt Q$, both $\mathcal M_K(t)$ and
its inverse are bounded by a constant depending only on $x$.

\subsection{Comparison with a constant-coefficient system}

Let $R_K$ be Kato parallel transport \cite{Kato1950,Kato}, defined by
\begin{equation*}
 R_K'=[P',P]R_K,\qquad R_K(0)=I.
\end{equation*}
Then $R_K$ is unitary and
\begin{equation*}
 R_K(t)^*P(t)R_K(t)=P_0:=P(0).
\end{equation*}

Put $\widehat W=R_K^*W$ and $\widehat D=R_K^*D_0R_K$.  Then
\begin{equation*}
 \widehat W'=\widehat K\widehat W,
 \qquad
 \widehat K=P_0\widehat DP_0+qP_0^\perp\widehat DP_0^\perp.
\end{equation*}
Let $\overline W$ solve the same equation after replacing the block on
$\operatorname{ran}P_0$ by $qP_0\widehat DP_0$.  The two block-diagonal
generators agree on $\operatorname{ran}P_0^\perp$, and hence
\begin{equation*}
 \rank(\widehat W(r)-\overline W(r))\le M,
 \qquad
 \|\widehat W(r)-\overline W(r)\|_{\Sone}\le2M.
\end{equation*}

We next compare $\overline W$ with $R_K^*e^{qD_0r}$.  The generator of
$R_K^*e^{qD_0r}$ differs from that of $\overline W$ by trace norm at most
\begin{equation*}
 2q\|P^\perp D_0P\|_{\Sone}+\|P'\|_{\Sone}.
\end{equation*}
Unitary Duhamel comparison and \cref{lem:P-bounds} therefore give
\begin{equation}\label{eq:W-frozen-prefix}
 \|W(r)-e^{qD_0r}\|_{\Sone}
 \le2M+CM\left(r+\frac r{\sqrt Q}\right).
\end{equation}

\begin{proposition}\label{prop:frozen-tail}
Let
\[
 \widehat J_K(t)=q^{1/2}\int_0^tW(r)\,dr,
 \qquad
 \widehat J_0(t)=q^{1/2}\int_0^te^{qD_0r}\,dr.
\]
Uniformly for $0<s\le x$,
\begin{equation}\label{eq:Jtail}
 \|\widehat J_K(s\sqrt Q)-\widehat J_0(s\sqrt Q)\|_{\Sone}
 \le C_x\bigl(Ms+M\sqrt Q\,s^2\bigr).
\end{equation}
\end{proposition}

\begin{proof}
Integrate \eqref{eq:W-frozen-prefix} and use $q^{1/2}=Q^{-1/2}$:
\begin{align*}
 \|\widehat J_K(s\sqrt Q)-\widehat J_0(s\sqrt Q)\|_{\Sone}
 &\le\frac1{\sqrt Q}\int_0^{s\sqrt Q}
 \left[2M+CM\left(r+\frac r{\sqrt Q}\right)\right]dr\\
 &\le C_x\bigl(Ms+M\sqrt Q\,s^2\bigr).
\end{align*}
\end{proof}

\subsection{A determinant comparison}\label{sec:gapfree}

We compare the Jacobi determinants using
\cite[Theorem~4.2]{FangFixed}.

\begin{theorem}\label{thm:output}
Fix $T>0$, let $\mathcal H$ be a finite-dimensional Hilbert space over
$\mathbb F$, and write $\mathcal M(t)=\mathcal M_0(t)R(t)$ on
$0\le t\le T$, where
\[
 R'=\widetilde E R,\qquad R(0)=I,
\]
and
\[
 \sup_{0\le t\le T}
 \bigl(\|\mathcal M_0(t)\|_{\op}+\|\mathcal M_0(t)^{-1}\|_{\op}\bigr)
 \le K_0.
\]
Let $\iota:\mathbb F^k\to\mathcal H$ be an isometric embedding,
$C:\mathcal H\to\mathbb F^k$ a bounded linear map, and set
\[
 B=C\mathcal M(T)\iota,\qquad B_K=C\mathcal M_0(T)\iota.
\]
If $B_0:\mathbb F^k\to\mathbb F^k$ is linear and satisfies
$\|B_K-B_0\|_{\Sone}\le\ell$ for some $\ell\ge0$, then, for every $\delta>0$,
\begin{equation}\label{eq:output-volume-comparison}
 \log\vol_k(B)
 \le \log\vol_{k,\delta}(B_0)
 +\eta_{\rm sym}+C_0\delta^{-1}(\eta+\ell),
\end{equation}
where
\[
 \eta=\int_0^T\|\widetilde E(t)\|_{\Sone}\,dt,
 \qquad
 \eta_{\rm sym}=\int_0^T\|\sym\widetilde E(t)\|_{\Sone}\,dt,
\]
and $C_0$ depends only on $K_0$ and $\|C\|_{\op}$.
\end{theorem}

For the present Jacobi system, let $\mathcal M$ and $\mathcal M_K$ be the
propagators of \eqref{eq:aug-actual} and \eqref{eq:Mref}, and put
\begin{equation*}
 R=\mathcal M_K^{-1}\mathcal M,
 \qquad \widetilde E=\mathcal M_K^{-1}\mathcal E\mathcal M_K.
\end{equation*}
For $t\le x\sqrt Q$,
\begin{equation}\label{eq:MK-uniform-bound}
 \|\mathcal M_K(t)\|_{\op}+\|\mathcal M_K(t)^{-1}\|_{\op}\le C_x,
\end{equation}
because $q^{1/2}t\le x$.  Use the input and output maps
\begin{equation*}
 \iota(v)=\binom0v,
 \qquad C\binom\zeta\xi=\zeta.
\end{equation*}
The Jacobi maps for the original, reference, and constant-coefficient
systems are $\widehat J$, $\widehat J_K$, and $\widehat J_0$, respectively.

\begin{proposition}
\label{prop:actual-frozen}
For every fixed $x>0$, uniformly for $0<s\le x$ and $0<\rho\le1$,
\begin{equation}\label{eq:actual-frozen}
 \log|\det\widehat J(s\sqrt Q)|
 \le\log\det_{\rho,s}\widehat J_0(s\sqrt Q)
 +C_x\left(\rho^{-1}MQ^{3/4}+MQ^{3/4}\right).
\end{equation}
\end{proposition}

\begin{proof}
By \eqref{eq:MK-uniform-bound}, \cref{prop:quarter-action} gives
\[
 \int_0^{s\sqrt Q}\!\|\widetilde E(t)\|_{\Sone}\,dt
 +\int_0^{s\sqrt Q}\!\|\sym\widetilde E(t)\|_{\Sone}\,dt
 \le C_xMQ^{3/4}s.
\]
Moreover, \cref{prop:frozen-tail} gives
\[
 \|\widehat J_K(s\sqrt Q)-\widehat J_0(s\sqrt Q)\|_{\Sone}
 \le C_x(Ms+M\sqrt Q\,s^2).
\]
Apply \cref{thm:output} with $T=s\sqrt Q$ and cutoff
$\delta=\rho s$.  Since $Q\ge1$ and $s\le x$, the regularization term is at
most $C_x\rho^{-1}MQ^{3/4}$, while $\eta_{\rm sym}\le C_xMQ^{3/4}$.
This proves \eqref{eq:actual-frozen}.
\end{proof}

\subsection{Polar integration}\label{sec:polar}

Let
\[
 S_0^{N-1}=\{X\in\mathfrak g:\|X\|_0=1\},
 \qquad u_0=A_Q^{-1/2}X,
\]
so that $u_0$ is metric unit, and write
$\Phi(t,X)=\operatorname{Exp}_{[I]}(tA_Q^{-1/2}X)$.

\begin{proposition}
\label{prop:polar-jacobian}
The metric-normalized differential of $X\mapsto\Phi(t,X)$ on the full
$N$-dimensional initial space is
\begin{equation}\label{eq:full-J}
 \mathcal J(t,X)=G(t)^{-1/4}\widehat J(t,X)G(0)^{-1/4}.
\end{equation}
Consequently
\begin{equation}\label{eq:detfull}
 |\det\mathcal J(t,X)|=Q^{(N-M)/2}|\det\widehat J(t,X)|.
\end{equation}
If $j_Q(t,X)$ denotes the normal polar Jacobian, then
\begin{equation}\label{eq:jq}
 j_Q(t,X)=Q^{(N-M)/2}\frac{|\det\widehat J(t,X)|}{t}.
\end{equation}
\end{proposition}

\begin{proof}
For a variation $h$ of $X$, the initial velocity variation is
$A_Q^{-1/2}h$ and the initial momentum variation is
$A_Q^{1/2}h=G(0)^{-1/2}h$.  Since $\xi=G^{1/4}\eta$,
\begin{equation*}
 \xi_0=G(0)^{-1/4}h.
\end{equation*}
At the endpoint, $z=G(t)^{1/4}\zeta$, and metric normalization in the spatial
frame multiplies $z$ by $G(t)^{-1/2}$.  Hence the endpoint coordinate is
$G(t)^{-1/4}\zeta$, proving \eqref{eq:full-J}.  Since
$\det G(t)=q^{N-M}=Q^{-(N-M)}$, each endpoint factor has determinant
$Q^{(N-M)/4}$, proving \eqref{eq:detfull}.

For the radial variation $h=X$, homogeneity of the exponential map gives the
Jacobi field $t\dot\gamma_X(t)$, of metric norm $t$.  The Gauss lemma makes it
orthogonal to all angular variations.  Thus the full volume equals $t$ times
the normal volume, proving \eqref{eq:jq}.
\end{proof}

Let $c(X)$ be the cut time of the corresponding metric-unit geodesic.  The
standard polar integration formula on the minimizing domain
\cite[Chapter~III]{Chavel} gives
\begin{equation*}
 \Vol_QB_Q([I],R)=\int_{S_0^{N-1}}\int_0^R
 \mathbf1_{\{t<c(X)\}}j_Q(t,X)\,dt\,d\varsigma_0(X).
\end{equation*}
The cut locus has measure zero.  Set $R=x\sqrt Q$ and $t=s\sqrt Q$.  Since
$dt/t=ds/s$ and $\varsigma_0(S_0^{N-1})=N\omega_N$,
\begin{equation*}
 \Vol_QB_Q([I],x\sqrt Q)
 =N\omega_NQ^{(N-M)/2}\int_0^x\E\left[
 \mathbf1_{\{s\sqrt Q<c(X)\}}|\det\widehat J(s\sqrt Q,X)|
 \right]\frac{ds}{s}.
\end{equation*}
On the other hand,
\begin{equation*}
 \Vol_Q\PU(D)=Q^{(N-M)/2}\Vol_0\PU(D).
\end{equation*}
All powers of $Q$ cancel:
\begin{equation}\label{eq:polar-normalized}
 \mu_D\bigl(B_Q([I],x\sqrt Q)\bigr)
 =\frac{N\omega_N}{\Vol_0\PU(D)}\int_0^x\E\left[
 \mathbf1_{\{s\sqrt Q<c(X)\}}|\det\widehat J(s\sqrt Q,X)|
 \right]\frac{ds}{s}.
\end{equation}

Let $H=-im_0/\sqrt Q$ have eigenvalues $h_1,\ldots,h_D$.  On the real Cartan
subspace, $D_0$ vanishes and $\widehat J_0(s\sqrt Q)=sI$.  On the real root
plane corresponding to $i<j$, $D_0$ has angular frequency
$\sqrt Q(h_i-h_j)$.  Thus, on a complex root vector with $D_0v=i\omega v$,
\begin{equation*}
 \widehat J_0(s\sqrt Q)v
 =q^{1/2}\frac{e^{iq\omega s\sqrt Q}-1}{iq\omega}v,
\end{equation*}
and the root singular value is
\begin{equation*}
 s\left|\sinc\frac{s(h_i-h_j)}2\right|.
\end{equation*}
Therefore
\begin{equation*}
 \det_{\rho,s}\widehat J_0(s\sqrt Q)
 =s^{D-1}\prod_{i<j}\left[
 s\max\left\{\left|\sinc\frac{s(h_i-h_j)}2\right|,\rho\right\}
 \right]^2.
\end{equation*}
The total real exponent is $(D-1)+2\binom D2=D^2-1=N$.

\section{Spectral reduction and small-ball estimates}\label{sec:transfer}

From this point on, $D=2^n$ is sufficiently large.  In particular,
$D-1>M$, which already holds for $n\ge9$, so the beta parameters below are
positive.  The comparison of spherical averages in
\cite[Lemmas~5.1--5.3 and Proposition~5.4]{FangFixed} is combined with the
low-weight momentum estimate in \cref{lem:cheap-removal}.

\subsection{Energy decomposition and regularized sinc products}

A uniform $X\in S_0^{N-1}\subset\mathfrak g$ decomposes as
\begin{equation*}
 X=\sqrt z\,X_{\cL}+\sqrt{1-z}\,X_{\cP},
\end{equation*}
where $X_{\cL}\in\cL$ and $X_{\cP}\in\cP$ are independent uniform unit
vectors and, independently,
\begin{equation*}
 z\sim\operatorname{Beta}\left(\frac M2,\frac{N-M}{2}\right).
\end{equation*}
Pass to the Hermitian realization by
\[
 H_{\cL}=-iX_{\cL}\in\mathfrak h_{\cL},
 \qquad
 H_{\cP}=-iX_{\cP}\in\mathfrak h_{\cP}.
\]
Since $u_0=A_Q^{-1/2}X$ and $m_0=A_Q^{1/2}X$, the normalized Hermitian
momentum $H=-im_0/\sqrt Q$ is
\begin{equation*}
 H=\sqrt{1-z}\,H_{\cP}+\frac{\sqrt z}{\sqrt Q}H_{\cL}.
\end{equation*}
Every Hilbert--Schmidt unit vector $H_{\cL}\in\mathfrak h_{\cL}$ satisfies
\begin{equation*}
 \|H_{\cL}\|_{\op}\le\sqrt M,
\end{equation*}
by expansion in the low-weight Pauli basis and Cauchy--Schwarz.

For a trace-zero Hermitian matrix $K$ with eigenvalues
$\kappa_1,\ldots,\kappa_D$, define
\begin{equation*}
 \mathcal R_{s,\rho}(K)
 :=s^{D-1}\prod_{i<j}\left[
 s\max\left\{\left|\sinc\frac{s(\kappa_i-\kappa_j)}2\right|,\rho\right\}
 \right]^2.
\end{equation*}
For a unit vector $K\in\mathfrak h$ and $0<y\le s$, write
\begin{equation*}
 \mathcal R_{s,y,\rho}(K)
 :=s^{D-1}\prod_{i<j}\left[
 s\max\left\{\left|\sinc\frac{y(\kappa_i-\kappa_j)}2\right|,\rho\right\}
 \right]^2.
\end{equation*}
For $L>0$ set
\begin{equation}\label{eq:F-transfer-definition}
 F_{s,y,\rho,L}(K)=
 \mathbf1_{\{y\diam\spec(K)\le L\}}\mathcal R_{s,y,\rho}(K).
\end{equation}

By \cite[Lemma~5.1]{FangFixed}, if $K,K'\in\mathfrak h$ are unit vectors,
$0<y\le s\le x_\star$, $0<\rho\le1$, and
$\|K-K'\|_{\op}\le\delta$ for some $\delta\ge0$, then
\begin{equation}\label{eq:root-log-lipschitz}
 \left|\log\mathcal R_{s,y,\rho}(K)
 -\log\mathcal R_{s,y,\rho}(K')\right|
 \le Cx_\star\rho^{-1}D^2\delta.
\end{equation}
The same eigenvalue perturbation bound gives, for every $L>0$,
\begin{equation}\label{eq:diameter-enlargement}
 y\diam\spec(K)\le L
 \quad\Longrightarrow\quad
 y\diam\spec(K')\le L+2x_\star\delta.
\end{equation}

\subsection{Spherical averages and the low-weight momentum}

Let $X$ and $Z$ be independent standard Gaussians in
$\mathfrak h_{\cP}$ and $\mathfrak h_{\cL}$, respectively, and put
\begin{equation*}
 \tau=e^{-\sqrt D},\qquad Y_\tau=X+\tau Z,\qquad Y_1=X+Z.
\end{equation*}
For a nonzero vector $V$, write $\widehat V=V/\|V\|_0$.  The density
comparison of \cite[Lemma~5.2]{FangFixed} gives, for every nonnegative
measurable $F$ on the full unit sphere,
\begin{equation*}
 \E F(\widehat Y_\tau)
 \le e^{M\sqrt D}\E F(\widehat Y_1).
\end{equation*}

Set $K_D=N-M$ and
\begin{equation*}
 a_D=e^{-\sqrt D/4},\qquad b_D=e^{\sqrt D/2},\qquad
 r_D=\frac{\tau b_D}{a_D\sqrt{K_D}},
\end{equation*}
\begin{equation*}
 \delta_D=\sqrt M\,r_D+\frac12\sqrt D\,r_D^2,
\end{equation*}
and
\begin{equation*}
 \mathcal G_D=
 \{\|X\|_0\ge a_D\sqrt{K_D},\ \|Z\|_0\le b_D\}.
\end{equation*}
By \cite[Lemma~5.3]{FangFixed},
\begin{equation*}
 \|\widehat X-\widehat Y_\tau\|_{\op}\le\delta_D,
 \qquad \delta_D=e^{-\Omega(\sqrt D)}
\end{equation*}
on $\mathcal G_D$, and
\begin{equation}\label{eq:bad-Gaussian-probability}
 \Prob(\mathcal G_D^c)
 \le e^{-cD^2\sqrt D}+e^{-ce^{\sqrt D}}.
\end{equation}

\begin{proposition}\label{prop:transfer}
For $0<y\le s\le x_\star$, $0<\rho\le1$, and $L>0$, define
\[
 I_{\cP}(s,y,\rho;L)=\E F_{s,y,\rho,L}(\widehat X),
 \qquad
 I_{\rm full}(s,y,\rho;L)=\E F_{s,y,\rho,L}(\widehat Y_1).
\]
Then
\begin{align}
 I_{\cP}(s,y,\rho;L)
 &\le e^{M\sqrt D+\varepsilon_D}
 I_{\rm full}(s,y,\rho;L+2x_\star\delta_D)\notag\\
 &\quad+s^N\left(e^{-cD^2\sqrt D}+e^{-ce^{\sqrt D}}\right),
 \label{eq:transfer}
\end{align}
where
\begin{equation}\label{eq:epsilonD-explicit}
 \varepsilon_D=Cx_\star\rho^{-1}D^2\delta_D.
\end{equation}
If $\rho\ge D^{-B}$ for fixed $B$, then $\varepsilon_D=o(1)$ uniformly.
\end{proposition}

\begin{proof}
This is \cite[Proposition~5.4]{FangFixed}, applied to the function
\eqref{eq:F-transfer-definition} with the present notation.
\end{proof}

\begin{lemma}
\label{lem:cheap-removal}
Let $0\le z<1$ and $L>0$.  Let $K_{\cP}$ and $K_{\cL}$ be Hilbert--Schmidt
unit vectors in the high-weight and low-weight Hermitian subspaces, and put
\begin{equation}\label{eq:Hz-definition}
 H_z=\sqrt{1-z}\,K_{\cP}+\frac{\sqrt z}{\sqrt Q}K_{\cL},
 \qquad y=s\sqrt{1-z}.
\end{equation}
Uniformly for $0<s\le x_\star$ and $0<\rho\le1$,
\begin{equation}\label{eq:cheap-determinant-comparison}
 \mathcal R_{s,\rho}(H_z)
 \le\exp\left\{Cx_\star\rho^{-1}D^2\sqrt{\frac MQ}\right\}
 \mathcal R_{s,y,\rho}(K_{\cP}).
\end{equation}
Furthermore,
\begin{equation}\label{eq:cheap-phase-enlargement}
 s\diam\spec(H_z)\le L
 \quad\Longrightarrow\quad
 y\diam\spec(K_{\cP})
 \le L+2x_\star\sqrt{\frac MQ}.
\end{equation}
\end{lemma}

\begin{proof}
Set $K'=\sqrt{1-z}K_{\cP}$.  Then
$\mathcal R_{s,\rho}(K')=\mathcal R_{s,y,\rho}(K_{\cP})$, while
\[
 \|H_z-K'\|_{\op}
 \le\frac{\sqrt z}{\sqrt Q}\|K_{\cL}\|_{\op}
 \le\sqrt{\frac MQ}.
\]
By Weyl's eigenvalue perturbation inequality, each eigenvalue difference
changes by at most $2\sqrt{M/Q}$.  The scalar function
\[
 t\longmapsto\log\max\{|\sinc(t/2)|,\rho\}
\]
is $C\rho^{-1}$-Lipschitz, as in the proof of
\cite[Lemma~5.1]{FangFixed}.  Summing the logarithmic changes over the
$2\binom D2$ real root factors and using $s\le x_\star$ gives
\eqref{eq:cheap-determinant-comparison}.  The spectral-diameter estimate
follows from
\[
 y\diam\spec(K_{\cP})=s\diam\spec(K')
 \le s\diam\spec(H_z)+2s\|H_z-K'\|_{\op}.
\]
\end{proof}

On the minimizing set, \cref{prop:phase-range} gives the initial phase
bound.  Lemma~\ref{lem:cheap-removal} and Proposition~\ref{prop:transfer}
then transfer the high-weight spherical average to the full Hermitian
sphere, enlarging the phase diameter to
\begin{equation*}
 2\pi+2x_\star\sqrt{\frac MQ}+2x_\star\delta_D.
\end{equation*}
Whenever $M/Q\to0$, this is at most
\begin{equation*}
 L_\star=2\pi+1
\end{equation*}
for all sufficiently large $D$.

\subsection{Weyl integration and the Abel estimate}\label{sec:weyl-abel}

We use the Weyl integration formula and Abel estimate of
\cite[Secs.~5.2--5.3]{FangFixed}, with the $Q$-dependence supplied by
\cref{prop:actual-frozen,lem:cheap-removal}.

Define the regularized chord kernel
\begin{equation*}
 K_\rho(t)=\max\{2|\sin(t/2)|,\rho|t|\}.
\end{equation*}

Let
\[
 \Sigma_D=\left\{\lambda\in\mathbb R^D:
 \sum_i\lambda_i=0,\ \sum_i\lambda_i^2=D\right\}
\]
with normalized surface probability $d\nu_D$, set $P_D=\binom D2$, and write
\[
 \Vdm(\lambda)=\prod_{1\le i<j\le D}(\lambda_i-\lambda_j),
 \qquad
 \mathcal Z_D=\int_{\Sigma_D}\Vdm(\lambda)^2\,d\nu_D(\lambda).
\]
For every nonnegative conjugation-invariant function $\Psi$ on the full
trace-zero Hermitian unit sphere, the normalized spherical Weyl formula
\cite{Hall,Mehta} is
\begin{equation}\label{eq:spherical-Weyl-formula}
 \E_{\rm full}\Psi(H)
 =\frac1{\mathcal Z_D}\int_{\Sigma_D}
 \Psi(\diag\lambda)\Vdm(\lambda)^2\,d\nu_D(\lambda).
\end{equation}

Set
\begin{equation}\label{eq:CD-definition}
 \mathfrak C_D=
 \frac{N\omega_N}{\Vol_0\PU(D)\,\mathcal Z_D},
 \qquad
 \mathfrak P_D=\frac{N\omega_N}{\Vol_0\PU(D)},
\end{equation}
and, for sufficiently large $D$,
\begin{equation*}
 \mathfrak B_D=
 \frac{B(M/2,(D-1-M)/2)}{B(M/2,(N-M)/2)}.
\end{equation*}
The evaluations in \cite[Sec.~5.3 and Eq.~(141)]{FangFixed} give
\begin{equation*}
 \log\mathfrak C_D=-D\log D+O(D),
 \qquad
 \log\mathfrak B_D
 =\frac M2\log D+O\left(M+\frac{M^2}{D}\right),
\end{equation*}
and
\begin{equation}\label{eq:polar-prefactor-crude}
 \log\mathfrak P_D=O(D^2\log D).
\end{equation}

For $0\le z\le1$ and $a\in\mathbb R$, set
\begin{equation*}
 \mathcal W_{\rho,s}(z,a)=
 \begin{cases}
 \displaystyle
 \frac{K_\rho(s\sqrt{1-z}\,a)}{\sqrt{1-z}\,|a|},
 &\sqrt{1-z}\,|a|>0,\\[3mm]
 s,&\sqrt{1-z}\,|a|=0.
 \end{cases}
\end{equation*}
Because $K_\rho(t)/|t|\to1$ as $t\to0$, this function is continuous.  With
$y=s\sqrt{1-z}$,
\begin{equation*}
 s\max\left\{\left|\sinc\frac{ya}{2}\right|,\rho\right\}
 =\mathcal W_{\rho,s}(z,a).
\end{equation*}
Hence
\begin{equation}\label{eq:weyl-cancel}
 \Vdm(\lambda)^2\prod_{i<j}\mathcal W_{\rho,s}
 (z,\lambda_i-\lambda_j)^2
 =(1-z)^{-P_D}\prod_{i<j}
 K_\rho(y(\lambda_i-\lambda_j))^2.
\end{equation}
The identity holds first for $z<1$ and distinct eigenvalues and then
extends by continuity.

The resulting chord product is controlled by the following Abel estimate.

\begin{lemma}\label{lem:abel}
Fix $L>2\pi$.  There are absolute constants $c_0,c_1,C>0$ with the
following property.  Suppose $0<\Delta\le1$ and real numbers
$\theta_1,\ldots,\theta_D$ satisfy
\begin{equation}\label{eq:abel-hyp}
 \max_{i,j}|\theta_i-\theta_j|\le L,
 \qquad
 \frac1D\sum_i\Arg(e^{i\theta_i})^2
 \le\frac{\pi^2}{3}-\Delta.
\end{equation}
For all sufficiently small $\Delta$, choose
\[
 h=c_0\Delta^2,\qquad r=e^{-h},\qquad \rho=(1-r)/L.
\]
Then
\begin{equation}\label{eq:abel-product}
 \log\prod_{i<j}K_\rho(\theta_i-\theta_j)^2
 \le-c_1\Delta^2D^2+CD\log(1/\Delta).
\end{equation}
For fixed $\Delta>0$, the same conclusion holds with fixed $\rho$ and an
$O_\Delta(D)$ remainder.
\end{lemma}

\begin{proof}
This is the lower-tail specialization of \cite[Lemma~5.6]{FangFixed}.
\end{proof}

\subsection{Uniform small-ball estimate}\label{sec:smallball}

\begin{lemma}\label{lem:bad-transfer}
Fix $x\le x_\star$ and $0<\rho\le1$.  Let
\[
 p_D=e^{-cD^{5/2}}+e^{-ce^{\sqrt D}}
\]
be the probability bound in \eqref{eq:bad-Gaussian-probability}.  After the
determinant comparison and removal of the low-weight component, the
contribution of $\mathcal G_D^c$ to the normalized polar integral is bounded by
\begin{align}
 \mathfrak T^{\rm bad}_{D,Q}(x,\rho)
 \le{}&\mathfrak P_D\frac{x^N}{N}
 \exp\left\{C_x\left(
 (1+\rho^{-1})MQ^{3/4}
 +\rho^{-1}D^2\sqrt{\frac MQ}\right)\right\}p_D.
 \label{eq:bad-transfer-bound}
\end{align}
Consequently
\begin{align}
 \log\mathfrak T^{\rm bad}_{D,Q}(x,\rho)
 \le{}&-cD^{5/2}+CD^2\log D\notag\\
 &+C_x\left(
 (1+\rho^{-1})MQ^{3/4}
 +\rho^{-1}D^2\sqrt{\frac MQ}\right).
 \label{eq:bad-transfer-log}
\end{align}
\end{lemma}

\begin{proof}
Start with the exact polar identity \eqref{eq:polar-normalized}.  The
determinant comparison contributes at most
\[
 \exp\{C_x(1+\rho^{-1})MQ^{3/4}\}.
\]
Conditioning on $z$ and removing the low-weight momentum contributes
at most
\[
 \exp\left\{C_x\rho^{-1}D^2\sqrt{M/Q}\right\}.
\]
The contribution from $\mathcal G_D^c$ is bounded by the second term in
\cref{prop:transfer}, using $F_{s,y,\rho,L}\le s^N$ and the probability
factor $p_D$.  The beta density integrates to one, and
\[
 \int_0^x s^N\frac{ds}{s}=\frac{x^N}{N}.
\]
The remaining normalized polar prefactor is exactly $\mathfrak P_D$.
This proves \eqref{eq:bad-transfer-bound}.  Now use
\eqref{eq:polar-prefactor-crude}, $N\log_+x=O(D^2)$ for fixed
$x\le x_\star$, and the first term in $p_D$; the second term is smaller.
Absorbing the lower-order terms into $C$ gives
\eqref{eq:bad-transfer-log}.
\end{proof}

\begin{proposition}[Uniform small-ball estimate]\label{prop:uniform-smallball}
There are absolute constants $\Delta_0,c,C>0$ such that, for every
sufficiently large power of two $D$, the following holds.  Assume
\begin{equation}\label{eq:smallball-hypotheses}
 D^{-1/4}\le\Delta\le\Delta_0,
 \qquad 0<x\le x_\star,
 \qquad x^2\le\frac{\pi^2}{3}-\Delta,
\end{equation}
and
\begin{equation}\label{eq:smallball-phase-condition}
 2x_\star\sqrt{M/Q}+2x_\star\delta_D\le1.
\end{equation}
Choose the Abel regularizer $\rho=\rho_\Delta\asymp\Delta^2$ from
\cref{lem:abel}.  Then
\begin{align}
 \mu_D\bigl(B_Q([I],x\sqrt Q)\bigr)
 &\le \mathfrak C_D\mathfrak B_D\frac{x^{D-1}}{D-1}
 \exp\{-c\Delta^2D^2+\mathcal E_{D,Q}(\Delta)\}
 +\mathfrak T^{\rm bad}_{D,Q}(x,\rho_\Delta),
 \label{eq:smallball-main}
\end{align}
where
\begin{equation}\label{eq:smallball-error}
 \mathcal E_{D,Q}(\Delta)
 \le C\left(
 \Delta^{-2}MQ^{3/4}+MQ^{3/4}
 +\Delta^{-2}D^2\sqrt{\frac{M}{Q}}
 +M\sqrt D+D\log(1/\Delta)
 \right)
\end{equation}
and
\begin{align}
 \log\mathfrak T^{\rm bad}_{D,Q}(x,\rho_\Delta)
 \le{}&-cD^{5/2}+CD^2\log D\notag\\
 &+C\left(
 (1+\Delta^{-2})MQ^{3/4}
 +\Delta^{-2}D^2\sqrt{\frac{M}{Q}}
 \right).
 \label{eq:bad-tail-explicit}
\end{align}
\end{proposition}

\begin{proof}
Start from the polar identity \eqref{eq:polar-normalized} and condition on
the beta variable $z$.  Proposition~\ref{prop:phase-range} restricts the
minimizing contribution to $s\diam\spec(H_z)<2\pi$.
Lemma~\ref{lem:cheap-removal} removes the low-weight component, and
Proposition~\ref{prop:transfer} transfers the resulting average to the full
Hermitian sphere.  Under \eqref{eq:smallball-phase-condition}, the transferred
phase diameter is at most $L_\star=2\pi+1$.

Apply the spherical Weyl formula \eqref{eq:spherical-Weyl-formula}.  With
$y=s\sqrt{1-z}$, the Hilbert--Schmidt normalization gives
\[
 \frac1D\sum_i\Arg(e^{iy\lambda_i})^2
 \le\frac1D\sum_i y^2\lambda_i^2=y^2\le s^2\le x^2
 \le\frac{\pi^2}{3}-\Delta.
\]
Thus \cref{lem:abel} applies.  The identity \eqref{eq:weyl-cancel} removes
the Vandermonde and leaves the factor $(1-z)^{-P_D}$.  Multiplying this by
the beta density and integrating gives
\begin{align*}
 &\frac1{B(M/2,(N-M)/2)}
 \int_0^1z^{M/2-1}(1-z)^{(N-M)/2-1-P_D}\,dz\\
 &\hspace{25mm}=
 \frac{B(M/2,(D-1-M)/2)}{B(M/2,(N-M)/2)}
 =\mathfrak B_D,
\end{align*}
because $(N-M)/2-P_D=(D-1-M)/2$.  After Weyl cancellation the only radial
power is the Cartan factor $s^{D-1}$, so
\[
 \int_0^x s^{D-1}\frac{ds}{s}=\frac{x^{D-1}}{D-1}.
\]
The polar/Weyl prefactor is $\mathfrak C_D$ by
\eqref{eq:CD-definition}.

The errors from the determinant comparison, removal of the low-weight
component, comparison of spherical averages, and Abel estimate are,
respectively,
\[
 \rho^{-1}MQ^{3/4}+MQ^{3/4},\qquad
 \rho^{-1}D^2\sqrt{M/Q},\qquad
 M\sqrt D+o(1),\qquad D\log(1/\Delta).
\]
Since $\rho^{-1}=O(\Delta^{-2})$, these give
\eqref{eq:smallball-error}.

The contribution from $\mathcal G_D^c$ is bounded by
\cref{lem:bad-transfer}.  With the Abel choice $\rho_\Delta\asymp\Delta^2$,
\eqref{eq:bad-transfer-log} gives \eqref{eq:bad-tail-explicit}.
\end{proof}

\section{Finite-height asymptotics and infinite-cliff consequences}\label{sec:consequences}

\subsection{Global upper bounds and circuit comparison}\label{sec:global-circuit}

\begin{lemma}\label{lem:phase-average}
For every $U\in U(D)$ there is a Hermitian $H$ satisfying
\begin{equation}\label{eq:H-projective}
 \Tr H=0,
 \qquad [e^{iH}]=[U],
 \qquad \|H\|_0^2\le\frac{\pi^2}{3}.
\end{equation}
\end{lemma}

\begin{proof}
Write the eigenvalues of $U$ as $e^{i\alpha_1},\ldots,e^{i\alpha_D}$.  For
$\phi\in[0,2\pi]$, let
\[
 \theta_j(\phi)=\Arg(e^{i(\alpha_j-\phi)})\in[-\pi,\pi].
\]
Each $\theta_j(\phi)$ is uniformly distributed on $[-\pi,\pi]$ when $\phi$
is uniform.  Hence
\[
 \frac1{2\pi}\int_0^{2\pi}\frac1D\sum_j\theta_j(\phi)^2\,d\phi
 =\frac1{2\pi}\int_{-\pi}^{\pi}t^2\,dt=\frac{\pi^2}{3}.
\]
Choose $\phi_0$ for which the integrand does not exceed its mean, set
$\bar\theta=D^{-1}\sum_j\theta_j(\phi_0)$, and let $H$ have eigenvalues
$\theta_j(\phi_0)-\bar\theta$ in an eigenbasis of $U$.  Then $H$ is traceless,
$e^{iH}$ differs from $U$ by a scalar phase, and
\[
 \|H\|_0^2=\frac1D\sum_j(\theta_j-\bar\theta)^2
 \le\frac1D\sum_j\theta_j^2\le\frac{\pi^2}{3}.
\]
\end{proof}

\begin{proposition}
\label{prop:global-upper}
For every $Q\ge1$ and every $[U]\in\PU(D)$,
\begin{equation}\label{eq:global-upper}
 d_Q([I],[U])\le x_\star\sqrt Q.
\end{equation}
Consequently $\diam(\PU(D),d_Q)\le x_\star\sqrt Q$.
\end{proposition}

\begin{proof}
Let $H$ be supplied by \cref{lem:phase-average}.  The path
$\gamma(t)=[e^{itH}]$ joins $[I]$ to $[U]$.  Since $A_Q\le QI$,
\[
 \ell_Q(\gamma)=\|iH\|_Q\le\sqrt Q\,\|H\|_0
 \le x_\star\sqrt Q.
\]
Right invariance gives the diameter estimate.
\end{proof}

\begin{lemma}[Projective approximation and circuit paths]
\label{lem:projective-circuit-interface}
For every $Q\ge1$ and $[U],[V]\in\PU(D)$,
\begin{equation}\label{eq:variable-projective-bridge}
 d_Q([U],[V])
 \le2\sqrt Q\,\arcsin\frac{\delta_{\op}([U],[V])}{2}.
\end{equation}
Every gate supported on at most two qubits has infinite-cliff distance at
most $x_\star$ from the identity, and every local layer has infinite-cliff
distance at most $x_\star\sqrt n$.  Consequently, if $V$ is the endpoint of
a no-ancilla circuit with $k$ arbitrary two-qubit gates and arbitrary
one-qubit gates, then
\begin{equation}\label{eq:variable-circuit-path}
 d_Q([I],[V])\le d_\infty([I],[V])
 \le x_\star(k+\sqrt n).
\end{equation}
\end{lemma}

\begin{proof}
Set $\delta=\delta_{\op}([U],[V])$ and
$\alpha=2\arcsin(\delta/2)$.  The centered principal-logarithm construction
in \cite[Proposition~8.1]{FangFixed} gives a traceless Hermitian matrix $K$
with $[V]=[Ue^{iK}]$ and $\|K\|_0\le\alpha$.  Since $A_Q\le QI$, the path
$[Ue^{itK}]$ has length at most $\sqrt Q\,\alpha$, which proves
\eqref{eq:variable-projective-bridge}.

Apply \cref{lem:phase-average} in dimension two or four and tensor the
resulting logarithm with identities on the other qubits.  This gives
horizontal length at most $x_\star$ for each one- or two-qubit gate.
The traceless one-qubit logarithms in a local layer are orthogonal for
$\langle\cdot,\cdot\rangle_0$, so the layer has horizontal length at most
$x_\star\sqrt n$.  Absorbing intermediate one-qubit gates into adjacent
arbitrary two-qubit gates leaves $k$ two-qubit gates and one terminal local layer, as
in \cite[Lemma~8.2]{FangFixed}.  Concatenation proves
\eqref{eq:variable-circuit-path}.
\end{proof}

\subsection{Finite-height asymptotics}

\begin{proof}[Proof of \cref{thm:variable-fixed}]
Fix $0<x<x_\star$ and set
\[
 \Delta_x=\min\left\{\Delta_0,\frac{\pi^2}{3}-x^2\right\}>0.
\]
For fixed $\Delta_x$, the first two assumptions in
\eqref{eq:variable-window-assumptions} imply that every $Q$-dependent term in
\eqref{eq:smallball-error} is $o(D^2)$.  Also
$M\sqrt D+D\log D=o(D^2)$, and the logarithms of the scalar factors are
$o(D^2)$.  Hence the strict Abel term $-c\Delta_x^2D^2$ dominates the main
term.  In \eqref{eq:bad-tail-explicit}, all terms after
$-cD^{5/2}$ are $o(D^{5/2})$.  The tail is therefore negligible, proving
\eqref{eq:variable-fixed-ball}.
\end{proof}

\begin{proof}[Proof of \cref{thm:variable-sharp}]
Let $\Xi=\Xi_{D,Q_D}$ and choose
\begin{equation*}
 \Delta=(A_{\mathrm{cal}}\Xi)^{1/4},
 \qquad x_{D,Q_D}=\sqrt{\frac{\pi^2}{3}-\Delta},
\end{equation*}
where $A_{\mathrm{cal}}$ is a sufficiently large absolute constant.  Since
$\Xi\ge(\log D)/D$, one has $\Delta\ge D^{-1/4}$ for all sufficiently
large $D$.  Since $\Xi\to0$, the upper bound on $\Delta$ and the phase
condition also hold.

With this choice, the Abel deficit is
$cA_{\mathrm{cal}}^{1/2}\Xi^{1/2}D^2$, whereas the regularization-amplified errors are at
most $CA_{\mathrm{cal}}^{-1/2}\Xi^{1/2}D^2$.  All remaining terms, including the scalar
factors, are $o(\Xi^{1/2}D^2)$.  Hence, for $A_{\mathrm{cal}}$ sufficiently large, the
negative term dominates.  The explicit tail is
$e^{-cD^{5/2}+o(D^{5/2})}$ and is smaller than the asserted failure bound, so
\[
 \mu_D(B_{Q_D}([I],x_{D,Q_D}\sqrt{Q_D}))
 \le e^{-c\Xi^{1/2}D^2}.
\]
Since $x_{D,Q_D}^2=x_\star^2-\Delta$,
\[
 x_\star-x_{D,Q_D}=\frac{\Delta}{x_\star+x_{D,Q_D}}
 \le C\Delta=CA_{\mathrm{cal}}^{1/4}\Xi^{1/4}.
\]
After increasing the absolute constant $C$ in
\eqref{eq:variable-concentration} if necessary, the event appearing there is
contained in $B_{Q_D}([I],x_{D,Q_D}\sqrt{Q_D})$.  This proves
\eqref{eq:variable-concentration}.  Together with
\cref{prop:global-upper}, it also proves convergence in probability.

For the $L^p$ convergence, set
\[
 X_D=\frac{d_{Q_D}([I],[U])}{\sqrt{Q_D}}.
\]
Then $0\le X_D\le x_\star$ deterministically.  For every fixed $p<\infty$
and every $\varepsilon>0$,
\[
 \E|X_D-x_\star|^p
 \le\varepsilon^p+x_\star^p
 \Prob\{|X_D-x_\star|>\varepsilon\}.
\]
First let $D\to\infty$ and then $\varepsilon\downarrow0$.

Finally, the lower-tail estimate supplies a point at distance at least
$x_{D,Q_D}\sqrt{Q_D}$, whereas \cref{prop:global-upper} bounds every distance from
above.  This proves \eqref{eq:variable-diameter}.
\end{proof}

\subsection{Critical-height consequences}

\begin{proposition}[Critical-height small balls at fixed subcritical radii]
\label{prop:critical-height}
For every fixed $0<x<x_\star$, there are constants $\kappa_x,b_x>0$ such
that, with
\begin{equation}\label{eq:critical-Q-fixed}
 Q_D=\kappa_x\frac{D^{8/3}}{M^{4/3}},
\end{equation}
one has
\begin{equation}\label{eq:critical-finite-ball}
 \mu_D\bigl(B_{Q_D}([I],x\sqrt{Q_D})\bigr)
 \le e^{-b_xD^2}
\end{equation}
for every sufficiently large power of two $D$.
\end{proposition}

\begin{proof}
Set
\[
 \Delta_x=\frac12\min\left\{\Delta_0,
 \frac{\pi^2}{3}-x^2\right\}>0.
\]
For $Q_D$ in \eqref{eq:critical-Q-fixed},
\begin{equation}\label{eq:critical-scaling-closed}
 MQ_D^{3/4}=\kappa_x^{3/4}D^2,
 \qquad
 D^2\sqrt{\frac{M}{Q_D}}
 =\kappa_x^{-1/2}D^{2/3}M^{7/6}=o(D^2).
\end{equation}
Choose $\kappa_x>0$ sufficiently small.  Then the terms proportional to
$MQ_D^{3/4}$ in \eqref{eq:smallball-error}, including the fixed factor
$\Delta_x^{-2}$, are absorbed by the Abel deficit $c\Delta_x^2D^2$.
The remaining error terms and the logarithm of the scalar prefactor in
\eqref{eq:smallball-main} are $o(D^2)$.  The phase condition also holds for
large $D$ because $M/Q_D\to0$ and $\delta_D\to0$.  After decreasing $b_x>0$,
the main contribution in \eqref{eq:smallball-main} is therefore at most
$e^{-b_xD^2}$.

For the contribution from $\mathcal G_D^c$, \eqref{eq:bad-tail-explicit} and
\eqref{eq:critical-scaling-closed} give
\[
 \log\mathfrak T^{\rm bad}_{D,Q_D}(x,\rho_x)
 \le-cD^{5/2}+O_x(D^2\log D),
\]
where $\rho_x\asymp\Delta_x^2$ is fixed.  Hence this contribution is
$e^{-\Omega(D^{5/2})}$ and can be absorbed into the same bound.
\end{proof}

\begin{proof}[Proof of \cref{thm:infinite-main}]
Fix $x_0=x_\star/2$, let $\kappa_{x_0},b_{x_0}$ be supplied by
\cref{prop:critical-height}, and use \eqref{eq:critical-Q-fixed} with
$x=x_0$.  By penalty monotonicity,
\[
 d_\infty\ge d_{Q_D},
 \qquad
 x_0\sqrt{Q_D}
 =x_0\sqrt{\kappa_{x_0}}\frac{D^{4/3}}{M^{2/3}}.
\]
Thus \eqref{eq:critical-finite-ball} proves
\eqref{eq:robust-infinite} with
$c_0=x_0\sqrt{\kappa_{x_0}}$ and $c_1=b_{x_0}$, and the probability estimate
immediately gives \eqref{eq:robust-diameter}.  Since $M=\Theta(n^2)$ and
$D=2^n$,
\[
 \log_2\left(D^{4/3}M^{-2/3}\right)
 =\frac43n-O(\log n),
\]
which proves \eqref{eq:exp-liminf}.  \end{proof}

\begin{proof}[Proof of \cref{cor:variable-circuit}]
Put
\[
 \alpha_\epsilon=2\arcsin(\epsilon/2),
 \qquad
 x=\frac12(\alpha_\epsilon+x_\star).
\]
The restriction $\epsilon<\epsilon_*$ is equivalent to
$\alpha_\epsilon<x_\star$, so $\alpha_\epsilon<x<x_\star$.  Let
$\kappa_x,b_x$ be supplied by \cref{prop:critical-height} and take $Q_D$ from
\eqref{eq:critical-Q-fixed}.

Let
\[
 [U]\notin B_{Q_D}([I],x\sqrt{Q_D}).
\]
Suppose that a no-ancilla circuit with $k$ arbitrary two-qubit gates has
endpoint $V$ satisfying $\delta_{\op}([U],[V])\le\epsilon$.  By
\cref{lem:projective-circuit-interface} and the triangle inequality,
\[
 x\sqrt{Q_D}
 \le d_{Q_D}([I],[U])
 \le x_\star(k+\sqrt n)+\alpha_\epsilon\sqrt{Q_D}.
\]
Consequently
\begin{equation*}
 k\ge
 \frac{x-\alpha_\epsilon}{x_\star}\sqrt{\kappa_x}
 \frac{D^{4/3}}{M^{2/3}}-\sqrt n.
\end{equation*}
Since
$\sqrt n=o(D^{4/3}M^{-2/3})$, the right-hand side is at least
$a_\epsilon D^{4/3}M^{-2/3}$ for all sufficiently large $D$, where
\[
 a_\epsilon=
 \frac{x-\alpha_\epsilon}{2x_\star}\sqrt{\kappa_x}>0.
\]
Thus the exceptional event in \eqref{eq:variable-circuit-lower} is contained
in $B_{Q_D}([I],x\sqrt{Q_D})$.  Equation~\eqref{eq:critical-finite-ball}
proves the claim with $b_\epsilon=b_x$.
\end{proof}

\begin{proof}[Proof of \cref{thm:infinite-refined}]
For the choice \eqref{eq:Qchoice-refined},
\[
 \frac{MQ_D^{3/4}}{D^2}=\Lambda_D^{-3/4},
 \qquad
 \sqrt{\frac{M}{Q_D}}
 =\frac{M^{7/6}\sqrt{\Lambda_D}}{D^{4/3}}=o(1).
\]
The remaining terms in \eqref{eq:Xi} are smaller because
$\Lambda_D=D^{o(1)}$.  Thus
$\Xi_{D,Q_D}=\Lambda_D^{-3/4}(1+o(1))$.  Apply
\cref{thm:variable-sharp} and use
\[
 \sqrt{Q_D}=\frac{D^{4/3}}{M^{2/3}\sqrt{\Lambda_D}},
 \qquad d_\infty\ge d_{Q_D}.
\]
The claimed probability bound follows.
\end{proof}

\subsection{Comparison with Brown's metric and the remaining exponent gap}\label{sec:brown-gap}

\begin{lemma}
\label{lem:brown-normalization}
Let $d_\infty^{\mathrm B}$ be Brown's infinite-cliff metric on $U(D)$ in
\cite[Eqs.~(1.3) and (3.2)]{BrownBG}, with the normalized trace
$\operatorname{tr}_D=D^{-1}\Tr$.  Let $\pi:U(D)\to\PU(D)$ be the quotient
map.  Then, for all $U,V\in U(D)$,
\begin{equation}\label{eq:Brown-quotient-contraction}
 d_\infty\bigl(\pi(U),\pi(V)\bigr)
 \le d_\infty^{\mathrm B}(U,V).
\end{equation}
Consequently
\begin{equation}\label{eq:Brown-diameter-comparison}
 \diam\bigl(U(D),d_\infty^{\mathrm B}\bigr)
 \ge\diam\bigl(\PU(D),d_\infty\bigr).
\end{equation}
If $\mu_D^{U}$ is normalized Haar measure on $U(D)$, then for every $r\ge0$,
\begin{equation}\label{eq:Brown-probability-comparison}
 \mu_D^{U}\{U:d_\infty^{\mathrm B}(I,U)\ge r\}
 \ge
 \mu_D\{[U]:d_\infty([I],[U])\ge r\}.
\end{equation}
\end{lemma}

\begin{proof}
For an absolutely continuous curve $U(t)$, put
$H(t)=i\dot U(t)U(t)^\dagger$, which is Hermitian, and expand it in the
Hermitian Pauli basis as
\[
 H(t)=h_0(t)I+\sum_{I\ne0}h_I(t)\sigma_I,
 \qquad h_I(t)=\operatorname{tr}_D(H(t)\sigma_I).
\]
Brown-admissibility means that $h_I(t)=0$ for every Pauli weight at
least three.  The projected right-trivialized velocity is
\[
 -i\bigl(H(t)-\operatorname{tr}_D(H(t))I\bigr)
 =-i\sum_{I\ne0}h_I(t)\sigma_I,
\]
and its squared horizontal norm is
\[
 \sum_{1\le\operatorname{wt}(I)\le2}|h_I(t)|^2.
\]
This is Brown's projective norm and does not exceed his norm when the scalar
term $|h_0(t)|^2$ is retained.  Hence $\pi$ sends every Brown-admissible curve
to a horizontal curve of no greater length.  This proves
\eqref{eq:Brown-quotient-contraction}, and taking suprema yields
\eqref{eq:Brown-diameter-comparison}.

The pushforward of Haar measure under $\pi$ is projective Haar measure.
Moreover,
\[
 \pi^{-1}\{[U]:d_\infty([I],[U])\ge r\}
 \subseteq\{U:d_\infty^{\mathrm B}(I,U)\ge r\}.
\]
Therefore \eqref{eq:Brown-probability-comparison} holds.
\end{proof}

Combining \cref{thm:infinite-main,lem:brown-normalization} shows that
Brown's infinite-cliff diameter has exponential lower rate at least $4/3$.
This contradicts the exponent-one conjecture in
\cite[Eq.~(3.2)]{BrownBG}.

The standard exact-synthesis bound gives the complementary upper estimate.

\begin{proposition}\label{prop:upper-infty}
There is an absolute constant $C$ such that
\begin{equation}\label{eq:upper-infty}
 \diam(\PU(D),d_\infty)\le CD^2.
\end{equation}
Consequently
\begin{equation}\label{eq:exponent-bracket}
 \frac43\le
 \liminf_{n\to\infty}\frac1n\log_2
 \diam\bigl(\PU(2^n),d_\infty\bigr)
 \le
 \limsup_{n\to\infty}\frac1n\log_2
 \diam\bigl(\PU(2^n),d_\infty\bigr)
 \le2.
\end{equation}
\end{proposition}

\begin{proof}
The gate-length part of \cref{lem:projective-circuit-interface} gives a
horizontal implementation of length at most $x_\star$ for every one- or
two-qubit gate, uniformly in $n$.  Quantum Shannon decomposition synthesizes
an arbitrary $n$-qubit unitary using $O(4^n)=O(D^2)$ CNOT and one-qubit gates
\cite{Shende}.  Concatenating the corresponding horizontal paths proves
\eqref{eq:upper-infty}.  Combining this estimate with
\eqref{eq:exp-liminf} gives \eqref{eq:exponent-bracket}.
\end{proof}

The lower rate $4/3$ comes from balancing the Jacobi-comparison error
\[
 \int_0^{O(\sqrt Q)}\|\mathcal E(t)\|_{\Sone}\,dt=O(MQ^{3/4})
\]
against the order-$D^2$ negative term in the Abel estimate.  Reducing this
error would allow larger cliff heights.  One possibility is to exploit
cancellation from the Hamiltonian or symplectic structure of the
off-diagonal Jacobi coupling.

\end{document}